\documentclass{article}
\usepackage{authblk}
\usepackage{amsmath}
\usepackage{array}
\usepackage{amsfonts}
\usepackage{amsmath}
\usepackage{rotate,epsfig}
\usepackage{caption}
\usepackage{amsfonts}
\usepackage{lscape}
\usepackage{wasysym}
\usepackage{graphicx}
\usepackage{color}
\usepackage{enumerate}
\usepackage[table]{xcolor}
\usepackage{url}
\definecolor{lightgray}{gray}{0.9}
\usepackage{longtable}
\usepackage[all]{xy}
\usepackage{lscape}

\usepackage{amsmath,amsthm,amssymb,graphicx,tikz,fancyhdr,hyperref}
\usepackage{graphicx,hyperref,fancyhdr,amsmath,amssymb,mathrsfs,color,xspace,pdflscape}
\usepackage[round]{natbib} 
\hypersetup{colorlinks=true,linkcolor=red,citecolor=blue,filecolor=magneta,urlcolor=cyan}
\DeclareMathSizes{1}{5}{5}{5}
\begin{document}
\title{ \bf \Large Statistical Inference on the Cumulative Distribution Function using Judgment Post Stratification}
\author{
\textbf{Mina Azizi Kouhanestani, \footnote{Department of Mathematical Sciences, Isfahan University of Technology,  Isfahan, 84105, Iran.  m.azizi@math.iut.ac.ir } \ \ Ehsan Zamanzade\footnote{ Department of Statistics, Faculty of Mathematics and Statistics, University of Isfahan, Isfahan 81746-73441, Iran. Corresponding author: e.zamanzade@sci.ui.ac.ir; ehsanzamanzadeh@yahoo.com}, \ \  and \ \ Sareh Goli \footnote{Department of Mathematical Sciences, Isfahan University of Technology,  Isfahan, 84105, Iran. s.goli@iut.ac.ir }}\\
}

\date {May 1, 2024}
\maketitle

\begin{abstract}
In this work, we discuss a general class of the estimators for the cumulative distribution function (CDF) based on judgment post stratification  (JPS)   sampling scheme which includes both empirical and kernel distribution functions.   Specifically, we obtain the expectation of the estimators in this class and show that they are asymptotically more efficient than their competitors in simple random sampling (SRS), as long as the rankings are better than random guessing. We find a mild condition that is necessary and sufficient for them to be asymptotically unbiased. We also prove that given the same condition, the estimators in this class are strongly uniformly consistent estimators of the true CDF,  and converge in distribution to a normal distribution when the sample size goes to infinity.
 We then focus on the kernel distribution function (KDF)  in the JPS design and  obtain the optimal bandwidth. We next carry out a comprehensive Monte Carlo simulation to compare the performance of the KDF in the JPS design for different choices of sample size, set size, ranking quality, parent distribution, kernel function as well as both perfect and imperfect rankings set-ups with its counterpart in SRS design. It is found that the JPS estimator dramatically improves the efficiency of the KDF as compared to its SRS competitor for a wide range of the settings. Finally, we apply the described procedure on a real dataset from medical context to show their usefulness and applicability in practice.

\end{abstract}

\newtheorem{theorem}{Theorem}
\newtheorem{acknowledgement}[theorem]{Acknowledgement}
\newtheorem{algorithm}[theorem]{Algorithm}
\newtheorem{axiom}[theorem]{Axiom}
\newtheorem{case}[theorem]{Case}
\newtheorem{claim}[theorem]{Claim}
\newtheorem{conclusion}[theorem]{Conclusion}
\newtheorem{condition}[theorem]{Condition}
\newtheorem{conjecture}[theorem]{Conjecture}
\newtheorem{corollary}[theorem]{Corollary}
\newtheorem{criterion}[theorem]{Criterion}
\newtheorem{definition}{Definition}
\newtheorem{example}[theorem]{Example}
\newtheorem{exercise}[theorem]{Exercise}
\newtheorem{lemma}[theorem]{Lemma}
\newtheorem{notation}[theorem]{Notation}
\newtheorem{problem}[theorem]{Problem}
\newtheorem{proposition}[theorem]{Proposition}
\newtheorem{remark}{Remark} 
\newtheorem{solution}[theorem]{Solution}
\newtheorem{summary}[theorem]{Summary}

\newpage
\section{Introduction}\label{Sec1}
In many practical studies including ecological, medical,  and environmental researches, there are cases in which  accurate quantification of characteristic  of interest (X) for all units in a sample is difficult or  costly but  ranking them can be done easily or inexpensively.  In such situations, two alternatives sampling designs to simple random sampling (SRS) are ranked set sampling (RSS) and judgment post stratification  (JPS)  sampling scheme.

RSS was  firstly suggested  by  \citet{mcintyre1952method} in his effort  for efficient estimation of pasture mean  and forage yields. He noticed that although measuring a yield is costly and tedious since it requires harvesting the crops,  an agri-expert can  produce a good ranking of the yields by  eye inspection.
To obtain an RSS sample of size $n$, one first determines the value of $H$ and a vector $ \mathbf{N}=(N_1, \ldots, N_H)$,  where $H$  is called the set size and $ \mathbf{N}$ is called the vector of post strata sample sizes such that   $N_r$ represents the count of units with rank $r$ that need to be chosen for accurate quantification, ensuring that the sum of all $N_r$ values from $r=1$ to $H$ is equal to $n$.
He/She then draws an SRS sample of size $n\times H$ from the population of interest and randomly divides it into $n$   samples of size $H$. In the next step,  each SRS sample of size $H$ is ranked in an increasing magnitude without actual measurement.
Lastly, from the $N_r$ ranked samples of size $H$, units with \textsf{judgment rank} $r$ are marked for accurate quantification ($r=1,  \ldots, H$). If $ N_1= \ldots=N_H$, sampling scheme is called balanced RSS.
The term \textsf{judgment rank} is employed to emphasize that the ranking process relies on personal judgment, visual assessment, or an auxiliary variable that strongly correlates with the variable of interest.
 Hence, it may not be precise and has some errors, which is called imperfect ranking. The best situation is perfect  ranking, in which there is no error in the mechanism of ranking and thus judgment ranks coincide with the true ones.

After the introduction of RSS, it was studied in many researches.  \citet{takahasi1968unbiased}  showed that the RSS mean estimator is not biased and has no larger variance than the SRS mean estimator of comparable size. 
Estimation of variance for  RSS  was considered by  \citet{stokes1980estimation}, \citet{maceachern2002new} and \citet{frey2013variance}.
\citet{stokes1988characterization} and \citet{dumbgen2020inference} studied that the piece-wise linear estimation of the cumulative distribution function (CDF) in RSS.
\citet{gulati2004smooth}  and \citet{eftekharian2017estimating} developed  kernel-type estimators of the CDF and investigated their asymptotic properties. Also  RSS method has been utilized to address almost all standard statistical topics such as the two sample problems \citep{mahdizadeh2017reliability,mahdizadeh2018interval,mahdizadeh2021smooth,frey2019omnibus,moon2020empirical,zamanzade2020efficient},   the prevalence estimation \citep{alvandi2021estimation,frey2019improved,frey2021robust} and estimation in a parametric family \citep{qian2021parameter,he2020maximum,he2021modified,chen2018global,chen2021pareto}.

\citet{maceachern2004judgement}  proposed JPS sampling plan as  a more flexible and applicable version of  RSS.
To obtain a JPS  sample of size $n$ using the set size $H$,  an SRS sample of size $n$ is drawn and  $X$ values for all units are quantified.  Then, for each quantified unit,   another SRS sample of size $H-1$  is drawn to create $n$ sets of size $H$ in total.
In the next step, each set of size $H$ is ranked in an increasing magnitude without actual measurement. Finally, in each set, the rank of the quantified unit is noted. 
The JPS data, therefore, includes $n$  independent and identically distributed (i.i.d) pairs $\{(X_i, R_i), i = 1, \ldots, n\}$, in which $X_i$  denotes the quantified value for the $i$th sampled unit and $R_i$ is its judgment rank among  $H$ units in the set. Thus, a JPS is a sampling plan in which an SRS sample is supported with judgment ranks.

 While RSS  and JPS sampling share a common underlying concept, there are significant distinctions between the two methods.
The first one is about when the mechanism of ranking is carried out. In RSS, the ranking process is done \textsf{prior} to selecting the units  for their exact  measurements, so the ranks are strongly connected to the observations and cannot be ignored. Thus, it is not possible to utilize SRS  statistical methods  for RSS samples. In fact the RSS samples can only be analyzed by appropriate techniques that are specifically developed for them. 
But  in JPS  sampling scheme, the ranks are obtained \textsf{after} quantifying sample units and therefore they are loosely related to them. So, JPS  sampling  is more flexible than RSS from practical point of view, since the existing SRS techniques can be readily used for JPS samples  by ignoring the ranks information. It is useful  when it is believed that the ranking quality is not good
 or suitable statistical techniques to analyze the JPS sample  are not available.
Another  difference is about the vector of post strata sample sizes: while the vector $\mathbf{N}$  is determined before the RSS sample is obtained, it is a random vector in JPS  sampling. This induces an extra variability to the JPS sample, so the statistical inference using JPS is anticipated to be slightly less efficient than RSS.

Recently, many studies have been conducted in JPS  sampling. For instance,  \cite{wang2008nonparametric} and \cite{frey2012improved} studied estimation of the population mean in JPS. 
\cite{ozturk2012combining} provided an alternative  JPS sampling plan to combine the judgment ranks of different rankers.
\cite{wang2012isotonized}  developed isotonized estimators for the CDF. The problem of variance estimation has been considered by \citet{frey2013variance} and \citet{zamanzade2016isotonized}.
\cite{ozturk2015statistical} developed sign test and quantile estimators for a JPS sample.
\citet{dastbaravarde2016some} discussed some  properties of nonparametric estimation in the JPS setting.
Estimation of population proportion for a JPS  sample was investigated by \citet{zamanzade2017estimation}.
\cite{zamanzade2018some} developed perfect judgment ranking for JPS  sampling  scheme.
\cite{omidvar2018judgment} combined  the judgment ranks and obtained the maximum likelihood estimation of model parameters.
\cite{ozturk2021judgment} presented some alternative JPS estimators to combine rank information from different sources. \cite{dizicheh2021efficient} studied odds estimation in the JPS and
\cite{alvandi2023analysis}  discussed estimation of categorical ordinal populations for a JPS sample

One deficiency concerns with the empirical distribution function  (EDF) is that it provides a discrete estimate for a continuous smooth CDF $F$. Thus, it is not possible to use its derivative to draw statistical inference about any functional of the probability density function (PDF). Moreover, it also fails to properly 
estimate the  CDF beyond the extreme observations. To address these concerns,  many studies have been done based on SRS for smooth estimation of the CDF, including \citet{nadaraya1964some}, \citet{watson1964hazard}, \citet{winter1973strong} and  \citet{yamato1973uniform}.

To the best of our knowledge, the problem of smooth estimation of the CDF has not been addressed in the JPS setting, yet. In Section \ref{Sec2}, we introduce a class of  estimators for the CDF in the JPS setting in a general form, which contains both empirical and smooth ones, and we study some of  their finite sample size properties.  We also prove that their asymptotic variances are  no larger than their SRS competitors of the same size regardless of ranking quality.  In Section \ref{Sec3}, we obtain some asymptotic results for our proposed estimators. Specifically, we   find the condition for them to   be   unbiased  and  strongly uniformly consistent estimators of the true CDF. We also establish their asymptotic normality  under the same condition. In Section \ref{Sec4}, we obtain some bias corrected results for the smooth CDF estimators and discuss optimal bandwidth selection with respect to minimizing the mean square error (MSE) of the estimators. Section \ref{Sec5}, contains an extensive simulation results to compare the performance of the JPS and SRS based  CDF estimators. The application of the proposed procedure is illustrated in Section \ref{Sec6}. The discussion section, labeled as Section  \ref{Sec7}, includes some final remarks and outlines potential avenues for future research. 

\section{Introduction of the CDF Estimator}\label{Sec2}
Let $X_1, \ldots, X_n$ be an SRS sample of size $n$ from a population with an unknown CDF $F$. A general class of the CDF estimators of $F$ is given by
\begin{eqnarray}\label{Fnsrs}
F_{n;srs}\left(t\right)=\frac{1}{n}\sum \limits_{i=1}^n K_n\left(t-X_i\right),
\end{eqnarray}
where $K_n\left(.\right)$ is a given CDF. 

Clearly,  $F_{n;srs}\left(t\right)$  is a convex combination of the CDFs, and thus it itself is a CDF, as well. In some practical situations, the CDF $F(t)$ is a continuous function in $t$, so it is reasonable to utilize a continuous CDF as  $K_n\left(.\right)$. Specially, if it is assumed that the CDF $F(t)$ is an absolutely continuous function in $t$, then a CDF estimator, which enjoys the said property, can be obtained by taking an absolutely continuous  CDF for $K_n\left(.\right)$. Thus an estimator for the population PDF $f$ can be obtained as 
\begin{eqnarray*}
f_{n;srs}\left(t\right)=\frac{1}{n}\sum \limits_{i=1}^n k_n\left(t-X_i\right),
\end{eqnarray*}
where $k_n\left(.\right)$ is the first derivative of $K_n\left(.\right)$.

Obviously, the estimator $F_{n;srs}(t)$ coincides with the EDF $F_{n;srs}^*(t)=\frac{1}{n}\sum \limits_{i=1}^n \mathbb{I}\left(X_i\leq t\right)$ if $K_n(t)=e(t)$, where $e(t)=\mathbb{I}\left(t\geq 0 \right)$, and $\mathbb{I}\left(.\right)$ is the standard indicator function. 

Statistical properties of   $F_{n;srs}\left(t\right)$   have been discussed by \citet{yamato1973uniform}. He found the necessary and sufficient condition for the estimators in this class to be asymptotically unbiased, and proved their consistencies and normalities under the same condition. 

Let $\{(X_i,R_i), i=1, \ldots, n\}$ is a JPS sample of size $n$ from a population with an unknown CDF $F$. As mentioned before,  $R_i$, for  $i=1, \ldots, n$ is judgment rank of $X_i$  among  $H$ units in its set, so,  we set $\bold{R} =(R_1, \ldots, R_n)$ as the ranks vector.  The variable $I_{ir}$ is defined as follows: If observation $X_i$ has a judgment rank of $r$ (i.e., $R_i = r$), then $I_{ir} = 1$; otherwise, $I_{ir} = 0$. This definition applies to every $i$ ranging from $1$ to $n$ and every $r$ ranging from $1$ to $H$.
 With this definition, if $N_r$  denotes the number of  observations $X_i$ in the $r$th post strata, then we have $N_r=\sum \limits_{i=1}^{n}I_{ir}$, and the vector of post strata sample sizes $\bold{N} = (N_1, \ldots, N_H)$ follows a multinomial distribution  with mass parameter $n$ and probability vector $(\frac{{1}}{H}, \ldots, \frac{{1}}{H})$. Likewise, if we  define $J_r=1/N_r$  if $N_r> 0$, otherwise $J_r= 0$, for $r = 1, \ldots, H$, then  $d_n=\sum \limits_{r=1}^{H} \mathbb{I}(N_r>0)$ denotes the number of nonempty post strata which formed during  JPS sampling procedure.
 Furthermore, note  that the conditional distribution of each  $X_i$ given its rank $R_i=r$, denoted as $F_{[r]}$, is equivalent  to the distribution of the $r$th order statistic ($X_{[r]}$) from a sample of size $H$.
  Finally, it is worth mentioning that if the same ranking process is applied for all sets of size $H$ during the course of JPS, then the ranking mechanism is called \textsf{consistent}. According to \citet{presnell1999u}, if the ranking process is consistent, then  for all  $t \in \mathbb{R},$ we have 
$$
F(t)=\frac{1}{H}\sum \limits_{r=1}^H F_{[r]}(t).
$$
The above equality, which is known as the \textsf{fundamental equality}, plays an important role in establishing our results.

Using above notation, a general class of the CDF estimators in the JPS setting is proposed as
\begin{eqnarray} 
F_{n;jps}(t) = \sum\limits_{r = {1}}^H {{W_{r}}}  {{F}}_{n;[r]} (t),\label{eq1} 
\end{eqnarray} 
where $ W_r=\frac{\mathbb{I}(N_r>0)}{d_n}$, 
${{F}}_{n;[r]} (t)=\frac{1}{N_r}\sum \limits_{i=1}^n K_n\left(t-X_i\right)I_{ir}$ if $N_r>0$, and zero otherwise, where $K_n\left(.\right)$ is a given CDF. Similar to SRS, an absolutely continuous estimator of the CDF in the JPS setting can be obtained by an appropriate choice of $K_n\left(.\right)$. Furthermore, if we take $K_n\left(t\right)=e(t)$, then the EDF in the JPS setting is obtained, which has the following form
\begin{eqnarray} 
F_{n;jps}^*(t) = \sum\limits_{r = {1}}^H {{W_{r}}}  {{F}}_{n;[r]}^* (t),\label{eq1j} 
\end{eqnarray} 
where 
$
{{F}}_{n;[r]}^* (t)=
\begin{cases}
\frac{1}{N_r}\sum \limits_{i=1}^n \mathbb{I}\left(X_i\leq t\right)I_{ir}  &   \text{if} \ N_r>0,
\\
0 & \text{otherwise.}
\end{cases}
$

To study the properties of $F_{n;jps}(t)$, we first note that both $F_{n;jps}(t)$ and $F_{n;srs}(t)$ have a common amount of bias. To see this, note that 
\begin{align} 
\mathbb{E}(F_{n;jps}(t))&=\mathbb {E}\left(\sum \limits_{r=1}^H W_r{{F}}_{n;[r]} (t)\right) \nonumber \\
&=\mathbb {E}\left(\mathbb {E}\left(\sum \limits_{r=1}^H W_r{{F}}_{n;[r]} (t)  \mid {\textbf{R}}\right)\right) \nonumber \\
&=\mathbb {E}\left( \sum \limits_{r=1}^H W_r \mathbb {E}\left(K_n(t-X_{[r]}) \right)\right) \nonumber \\
&=\mathbb{E}(W_1) \sum \limits_{r=1}^H  \mathbb {E}\left(K_n(t-X_{[r]}) \right)\nonumber \\
&=\frac{1}{H} \sum \limits_{r=1}^H  \mathbb {E}\left(K_n(t-X_{[r]}) \right) \nonumber \\
&=\mathbb{E}(F_{n;srs}(t)), \label{eqexp}
\end{align}
in which the second last equality follows from the fact that $W_1, \ldots, W_H$ are identically distributed with  $\mathbb {E}(W_1)=\frac{1}{H}$, and the last equality is an immediate result of the fundamental equality. 

To obtain  the variance of $F_{n;jps}(t)$, by following lines of \citet{dastbaravarde2016some}, and relying on the conditional variance formula, one can obtain
\begin{align} 
\mathbb {V}(F_{n;jps}(t))&= &\nonumber \\
&  \mathbb {E}(W_1^2 J_1) \sum \limits_{r=1}^H  \mathbb {V}\left(K_n(t-X_{[r]})\right) + \frac{H}{H-1} \mathbb {V}(W_1) \left(\sum\limits_{r=1}^{H} {\left({\mathbb{E} \left(K_n(t-X_{[r]})\right) - \mathbb{E} \left(K_n(t-X)\right)}\right)}^2 \right).\label{eqvar1}
\end{align} 
Using the fundamental equality, one can show that
 $\mathbb {V}(K_n(t-X))=\frac{1}{H} \sum \limits_{r=1}^H  \mathbb {V}\left(K_n(t-X_{[r]})\right) +\frac{1}{H}    \sum\limits_{r=1}^{H} {\left({\mathbb{E} \left(K_n(t-X_{[r]})\right) - \mathbb{E} \left(K_n(t-X)\right)}\right)}^2$,
and therefore, the variance of $F_{n;jps}(t)$ can be also re-written as 
\begin{align} 
\mathbb {V}(F_{n;jps}(t))&=H \mathbb {E}(W_1^2 J_1)  \mathbb {V}\left(K_n(t-X)\right) \nonumber \\ 
&-\left[  \mathbb {E}(W_1^2 J_1)  - \frac{H}{H-1} \mathbb {V}(W_1)\right]\left(  \sum\limits_{r=1}^{H} {\left({\mathbb{E} \left(K_n(t-X_{[r]})\right) - \mathbb{E} \left(K_n(t-X)\right)}\right)}^2  \right).\label{eqvar2}
\end{align} 

The properties of the vector $\left(W_1, \ldots, W_H\right)$ have been investigated by \citet{dastbaravarde2016some}. Specifically, they showed that for $r=1, \ldots, H$
\begin{flalign} 
\mathbb {V}(W_r)&=\frac{1}{H^2}\sum\limits_{l=1}^{H-1} {(\frac{l}{H})}^{n-1},&& \label{eq6}
\end{flalign} 
\begin{flalign} 
\mathbb {E}(W_r^2 J_r) &= \frac{1}{H^n} \left[\frac{1}{n} + \sum\limits_{d_n=2}^{H} \sum\limits_{j=1}^{d_n-1}  \sum\limits_{n_1=1}^{n-d_n+1}  \frac{{(-1)}^{j-1}}{d_n^2 n_1} {H-1 \choose d_n-1}  {d_n-1 \choose j-1} {n \choose n_1} {(d_n-j)}^{n-n_1} \right].&& 
\label{eq7}
\end{flalign} 
Also  for every  fixed $H$, they have established a proof that as the sample size $n$ approaches infinity, then
\begin{flalign} 
n\mathbb {V}(W_r) \to 0,&& \label{eq8.1}
\end{flalign} 
\begin{flalign} 
nH \mathbb {E}(W_r^2 J_r) \to 1,&& \label{eq8.2}
\end{flalign} 
\begin{flalign} 
\sqrt n \left( W_r-\frac{1}{H}\right) \xrightarrow[]{\text{a.s}} {0},&& \label{eq9}
\end{flalign}
where $\xrightarrow[]{\text{a.s}}$ means the almost sure convergence. 
 
According to \eqref{eq8.1} and  \eqref{eq8.2}, we can write
$$\lim_{n \to \infty }\mathbb {V}\left(F_{n;jps}(t)\right) = \mathbb {V}\left(F_{n;srs}(t)\right)\left(1-\delta\right),$$
where $\delta=\frac{1}{H\mathbb{V}(K_n(t-X))} \sum\limits_{r=1}^{H} {\left({\mathbb{E} \left(K_n(t-X_{[r]})\right) - \mathbb{E} \left(K_n(t-X)\right)}\right)}^2 $ is less than one.
Thus, we can conclude that asymptotic variance of the CDF estimator in JPS setting  is no larger than its SRS competitor of size $n$, even when the ranking is not perfect but still  better than random guessing.

\section{Some Asymptotic Results}\label{Sec3}
In this section, we study the asymptotic properties of the CDF estimators in the proposed class, when the set size $H$ is fixed and the sample size $n$ tends to infinity. We prove Glivenko–Cantelli type convergence for the estimators  and establish the asymptotic distribution of them.
To do so, we rely on the following definition for convergence of a sequence of CDFs. 

\begin{definition}
Let $\{G_n\}$ be a sequence of CDFs and $G$ be a CDF. Then we  define the sequence $\{G_n\}$ weakly converges to $G$ and write 
 ${G_n} \xrightarrow[]{w} G$, if 
  $${G_n}(t) \to G(t),\ \ \    \ \  \ \text{as} \ \ \ n \rightarrow \infty,$$
  $\text{for all}\  t \in C(G)$, where $C(G)$ denotes the set of all continuity points of $G$.
\end{definition}

The first theorem concerns with asymptotic mean  of the JPS estimators of the CDF. Specifically, it finds necessary and sufficient condition for asymptotic unbiasedness of   $F_{n;jps}(t)$.  Theorem \ref{mean} can be simply proven using Lemma 1 in \citet{yamato1973uniform} and equation \eqref{eqexp}, and thus its proof is omitted.

\begin{theorem} \textbf{(Asymptotic unbiasedness)}\label{mean}
Let  $\{(X_i,R_i),  i={1},\dots, n\} $ be a JPS sample from a population with  an unknown CDF $F$. Then if the mechanism of ranking is consistent, and for a fixed value of $H$, every estimator of the form $F_{n;jps}(t)$ in \eqref{eq1} is asymptotically unbiased if and only if  ${K_n} \xrightarrow[]{w}  e$.
\end{theorem}

The next result concerns with Glivenko-Cantelli type convergence property of the  estimator $F_{n;jps}(t)$.
\begin{theorem} \textbf{(Strongly uniformly convergence)}
Let
 $\{(X_i,R_i),  i={1},\dots, n\} $
be a JPS sample   from a population with an unknown CDF $F$.
If the mechanism of ranking is consistent and ${K_n} \xrightarrow[]{w}  e$, then for a fixed value of $H$, we have
$$\sup_{t \in\mathbb{R} }\big \lvert  F_{n;jps}(t) - F(t) \big \rvert \xrightarrow[]{\text{a.s}} {0}.$$
\end{theorem}

\begin{proof}
From fundamental equality, we can write
\begin{align}
\sup_{t \in\mathbb{R} }\big \lvert  F_{n;jps}(t) - F(t) \big \rvert &= \sup_{t \in\mathbb{R} }\big \lvert  \sum\limits_{r = {1}}^H {{W_{r}}}  F_{n;[r]}(t) - \frac{1}{H} \sum\limits_{r = {1}}^H {{F_{{[r]}}}(t) } \big \rvert  \nonumber \\
&\le \sum\limits_{r = {1}}^H \sup_{t \in\mathbb{R} } \big \lvert {{W_{r}}}  F_{n;[r]}(t) - \frac{1}{H} {{F_{{[r]}}}(t) } \big \rvert, \label{eq10}
\end{align}
where the last inequality follows from  triangle inequality. Besides, for every
 $r={1}, \ldots, H$, we have
\begin{align*}
\sup_{t \in\mathbb{R} } \big \lvert {{W_{r}}}  F_{n;[r]}(t) - \frac{1}{H} {{F_{{[r]}}}(t) } \big \rvert &\le \sup_{t \in\mathbb{R} } \big \lvert {{W_{r}}}  F_{n;[r]}(t) - \frac{1}{H} {{ F_{n;{[r]}}}(t) }\big \rvert  + \sup_{t \in\mathbb{R} } \big \lvert {{ \frac{1}{H} }}  F_{n;[r]}(t) - \frac{1}{H} {{F_{{[r]}}}(t) } \big \rvert    \nonumber \\
&=\sup_{t \in\mathbb{R} } \left(\big \lvert {{W_{r}}-\frac{1}{H} } \big \rvert  F_{n;[r]}(t) \right) +   \sup_{t \in\mathbb{R} } \left(\frac{1}{H}\big \lvert   F_{n;[r]}(t) - F_{[r]}(t)  \big \rvert \right) \nonumber \\
&=\big \lvert {{W_{r}}-\frac{1}{H} } \big \rvert   \sup_{t \in\mathbb{R} }   F_{n;[r]}(t) +  \frac{1}{H}  \sup_{t \in\mathbb{R} }\big \lvert   F_{n;[r]}(t)- F_{[r]}(t)  \big \rvert. \nonumber 
\end{align*}
It is clear from equation \eqref{eq9} we have 
$ \left( {{W_{r}} - \frac{{1}}{H}} \right) \xrightarrow[]{\text{a.s}}{0}$. Furthermore, 
from Theorem 3 in \citet{yamato1973uniform}, one can conclude that
 $$\sup_{t \in\mathbb{R} } \big \lvert   F_{n;[r]}(t)- F_{[r]}(t)  \big \rvert\xrightarrow[]{\text{a.s}}{0},$$
for $r=1, \ldots, H$, and this completes the proof.
\end{proof}

We now establish the asymptotic normality of the JPS estimator. 
\begin{theorem}\textbf{(Asymptotic normality)} 
Let
 $\{({X_i},{R_i}),i={1},\dots, n \}$ be a JPS   sample  from a population with an unknown CDF $F$.  If the mechanism of ranking is  consistent and ${K_n} \xrightarrow[]{w}  e$, then for a fixed value of $H$, we have
 $$\sqrt{n} \left( F_{n;jps}(t)-F(t)\right)\xrightarrow[]{\text{d}} N\left( {{0},\frac{{1}}{H}\sum\limits_{r = {1}}^H {{F_{[r]}}(t)[{1} - } {F_{[r]}}(t)]} \right),$$
 for all $t\in C(F)$ with $F(t)\ne 0$ or $1$, where $\xrightarrow[]{\text{d}}$ means convergence in distribution.
 \end{theorem}
 
 \begin{proof}
We can write
\begin{align*}
\sqrt{n} \left( F_{n;jps}(t)-F(t)\right)&=\sqrt{n}\left( \sum\limits_{r = {1}}^H {{W_r}}  F_{n;[r]}(t)-F(t) \right) \nonumber \\
&=\sqrt{n} \left(\sum\limits_{r = {1}}^H {W_r \left( F_{n;[r]}(t)-F_{[r]}(t) \right)} \right)
+ \sqrt{n} \left(\sum\limits_{r = {1}}^H {F_{[r]}(t)}\left({W_r}-\frac{{1}}{H}\right)\right).
\end{align*}
From \eqref{eq9}, it is clear that $\sqrt{n} \left(\sum\limits_{r = {1}}^H {F_{[r]}(t)} \left({W_r} - \frac{{1}}{H}\right)\right)\xrightarrow[]{\text{a.s}}{0}$, thus it is sufficient to obtain asymptotic distribution of 
$\sqrt{n} \left(\sum\limits_{r = {1}}^H {W_r \left( F_{n;[r]}(t) - F_{[r]}(t) \right)} \right)$.
For this purpose, we define
$$\sqrt{n} \left( \sum\limits_{r = {1}}^H {W_r \left( F_{n;[r]}(t)-F_{[r]}(t)\right)}\right) = \mathbf{A_n T_n},$$
 where  $\mathbf{A_n} =\sqrt{n}\left(\sqrt{J_1}W_1, \ldots, \sqrt{J_H}W_H\right) $
and $$\mathbf{T_n}={\left(\sqrt{N_1}\left(  F_{n;[1]}(t)-F_{[1]}(t)\right), \ldots, \sqrt{N_H}\left( F_{n;[H]}(t)-F_{[H]}(t)\right)\right)}^T,$$
  where $T$ denotes the transpose operator. Also from \eqref{eq9}, we have $ \left( {{W_r} - \frac{{1}}{H}} \right) \xrightarrow[]{\text{a.s}}{0}$ and  using law of large numbers, we obtain $nJ_r\xrightarrow[]{\text{a.s}}{H}$. Thus $\mathbf{A_n} \xrightarrow[]{\text{a.s}} {\left(\sqrt{\frac{1}{H}},\ldots,\sqrt{\frac{1}{H}}\right)}$.
\\
On the other hand, since $ F_{n;[r]}(t)$s,  given the vector of ranks $\textbf{R}$, are conditional independent random variables, we can write
\begin{align*}
P(\displaystyle  \bigcap _{r=1}^{H} \sqrt{N_r} \left( F_{n;[r]}(t)-F_{[r]}(t))\le t_r \right) 
&=E\left[P(\displaystyle  \bigcap _{r=1}^{H} \sqrt{N_r} \left( F_{n;[r]}(t)-F_{[r]}(t)\right)\le t_r \mid \textbf{R}) \right]\nonumber \\
&=E\left(\prod_{r=1}^H P\left( \sqrt{N_r}( F_{n;[r]}(t)-F_{[r]}(t))\le t_r \mid \textbf {R}\right) \right)\nonumber \\
&\xrightarrow[]{} E\left(\prod_{r=1}^H P\left( T_r \le t_r\right)\right)\nonumber \\
&=P\left(\displaystyle  \bigcap _{r=1}^{H} T_r \le t_r\right).
\end{align*}
From Theorem 1 in \citet{yamato1973uniform}, one can conclude that if ${K_n} \xrightarrow[]{w}  e$, then $T_r =\sqrt{n_r}\left(F_{n;[r]}(t)-F_{[r]}(t)\right)$  follows a mean-zero normal distribution with variance $F_{[r]}(t)\left[1-F_{[r]}(t)\right]$, for all $t\in C(F_{[r]})$ with $F_{[r]}(t)\ne 0$ or $1$.

Thus, we can conclude that the vector $\mathbf{T_n}$  converges in distribution to an H-dimensional normal distribution with zero mean vector and
variance covariance matrix $\pmb{\sum}$, where $\pmb{\sum}$  is a diagonal matrix that its $r$th diagonal element is given by $F_{[r]}(t)\left[1-F_{[r]}(t)\right]$ for $r=1,\ldots, H$. So, using Slutsky’s theorem, we have
$$\sqrt{n}\left(\sum\limits_{r = {1}}^H {W_r \left( F_{n;[r]}(t)-F_{[r]}(t)\right)}\right) =\mathbf{A_n T_n} \xrightarrow[]{\text{d}} N\left({0},\frac{{1}}{H}\sum\limits_{r = {1}}^H {{F_{[r]}}(t)\left[1 - {F_{[r]}}(t)\right]} \right).  $$
 \end{proof}

\section{CDF Estimation using Kernel Function}\label{Sec4}
Up to now, we assume that the function $K_n\left(.\right)$ is an arbitrary CDF and all asymptotic results obtained under the mild condition that $K_n \xrightarrow[]{w} e$. However, in the statistical literature, many researchers prefer to impose some more constraints on $K_n\left(.\right)$,  such that the resulting CDF estimator enjoys  some nice properties like smoothness and having less bias. In this section, we impose such constraints on proposed CDF estimator in the JPS setting and then we obtain the optimal estimator in this case. In so doing, we assume that the  CDF $K_n\left(t\right)$ depends on the sample size $n$ via the parameter $h_n$, and therefore, it can be written as $K_n\left(t\right)=K\left(\frac{t}{h_n}\right)$, where $K\left(.\right)$ is a given CDF and the parameter $h_n$ is an smoothing parameter and is called \textsf{bandwidth}. We  assume that the given CDF function $K\left(.\right)$ can be written as
\begin{align}\label{K.CDF}
K(t)=
\begin{cases}
0 &   t<-a,
\\
\int_{-a}^{t} k(x)dx & \lvert t  \rvert \le a,
\\
1 & t>a,
\end{cases}
\end{align} 
such that
$k(.)$ is a symmetric PDF with a bounded support on $\left[-a, a\right]$, and thus it satisfies  the following constraints
\begin{align} \label{ker.pdf}
\int_{-a}^{a} k(x)dx =1,\ \ \ \ 
\int_{-a}^{a} xk(x)dx =0,\ \ \ \ 
\int_{-a}^{a} x^2 k(x)dx \ne 0.   
\end{align} 
It is also assumed that the bandwidth  $h_n$ follows the conditions
\begin{eqnarray*}
\lim_{n \to \infty} h_n=0,\ \ \ \ \lim_{n \to \infty} nh_n=\infty,
\end{eqnarray*}
indicating that it converges to zero but  at a rate slower than $n^{-1}$. To make the notations less cumbersome, hereafter, we omit the subscript $n$ in $h_n$.  In the statistical literature, the function  $k\left(.\right)$ satisfying \eqref{ker.pdf} is usually known as \textsf{kernel} function and the method for estimating CDF (PDF) based on the function $k\left(.\right)$ is called kernel CDF (PDF) estimation. Therefore, the kernel distribution function (KDF) based on an SRS sample is obtained by replacing  \eqref{K.CDF} in equation  \eqref{Fnsrs} as 
\begin{eqnarray}
F_{n;srs}^k\left(t\right)=\frac{1}{n}\sum\limits_{i=1}^n K\left(\dfrac{t-X_i}{h}\right).
\end{eqnarray}

\citet{azzalini1981note}   obtained the optimal value for bandwidth $h$ based on an SRS sample of size $n$ by minimizing large sample MSE of $F_{n;srs}^k\left(t\right)$ as 
\begin{eqnarray*}
h_{srs}={\left(\frac{f(t)\left(a-\int_{-a}^{a} K^2(x)dx\right)}{{n\left( f^{\prime}(t) \int_{-a}^{a} x^2k(x)dx\right)}^2}\right)}^{\tfrac{1}{3}}.
 \label{hoptsrs}
\end{eqnarray*}

Let $\{(X_i,R_i), i=1, \ldots, n\}$ is a JPS sample of size $n$ from a population with an unknown CDF $F$. The kernel-type  estimator of CDF based on the JPS sample can be obtained by substituting \eqref{K.CDF} in equation \eqref{eq1}, and  is given by
\begin{eqnarray*}
F^{k}_{n;jps}\left(t\right)=\sum \limits_{r=1}^H W_r F^{k}_{n;[r]} (t),
\end{eqnarray*}
where ${F}^{k}_{n;[r]} (t)=\frac{1}{N_r}\sum \limits_{i=1}^n K\left(\frac{t-X_i}{h}\right) I_{ir}$ if $N_r>0$, and zero otherwise.

It is worth mentioning that the asymptotic results established in Section \ref{Sec3} hold for kernel-type estimators of the CDF, since ${K} \xrightarrow[]{w}  e$ as the sample size $n$ goes to infinity.

To obtain the mean of a KDF in the JPS setting,  we first note that  according to \eqref{eqexp}, one can write
\begin{align} 
\mathbb{E}\left(F^k_{n;jps}(t)\right) &=  \mathbb {E} \left(K\left(\frac{t-X}{h}\right) \right) \nonumber \\
&=\int_{-\infty}^{\infty} K\left(\frac{t-x}{h}\right)f(x)dx \nonumber \\
&=\int_{-\infty}^{t-ah} f(x)dx + \int_{t-ah}^{t+ah} K\left(\frac{t-x}{h}\right)f(x)dx.\nonumber
\end{align}
Note that since $k(.)$ is symmetric around 0, we can write $K(x)=\frac{1}{2} + g(x);\ \lvert x  \rvert \le a $, where g(.) is an odd function, and therefore we have
\begin{align} 
\mathbb{E}\left(F^k_{n;jps}(t)\right) &= \int_{-\infty}^{t-ah} f(x)dx  + \int_{t-ah}^{t+ah} \left(\frac{1}{2}+g\left(\frac{t-x}{h}\right)\right)f(x)dx\nonumber \\
&=\frac{1}{2} \left(F(t+ah)+F(t-ah)\right) +  \int_{t-ah}^{t+ah} g\left(\frac{t-x}{h}\right)f(x)dx \nonumber \\
&=\frac{1}{2} \left(F(t+ah)+F(t-ah)\right) + h  \int_{-a}^{a} g\left(u\right)f(t-hu)du.\nonumber
\end{align}
Using Taylor series of $F(t+ah)$, $F(t-ah)$ and $f(t-hu)$ around $t$, we have
\begin{align} 
\mathbb{E}\left(F^k_{n;jps}(t)\right) &= F(t)+\frac{a^2h^2}{2} f^\prime (t)+hf(t)\int_{-a}^{a} g(u)du-h^2f^\prime (t) \int_{-a}^{a} ug(u)du + O(h^3)\nonumber \\
&= F(t)+\frac{a^2h^2}{2} f^\prime (t)-h^2f^\prime (t) \int_{-a}^{a} ug(u)du + O(h^3)\nonumber \\
&=F(t)+h^2 f^\prime (t) \left(\frac{a^2}{2}-\int_{-a}^{a} ug(u)du\right)+ O(h^3),\nonumber 
\end{align}
where the second last equality follows from the fact that $g(.)$ is an odd function.

By using integration by parts, we can write
$$\int_{-a}^{a} ug(u)du =\frac{a^2}{2}-\frac{1}{2} \int_{-a}^{a} u^2k(u)du.$$
Therefore, we have
\begin{align} 
\mathbb{E}\left(F^k_{n;jps}(t)\right) &=F(t)+\frac{h^2}{2} f^\prime (t) \int_{-a}^{a} u^2k(u)du +O(h^3)\label{expkerO3} \\
&=F(t)+O(h^2).\label{expkerO2}
\end{align}
To obtain the variance of $F^k_{n;jps}(t)$, we first require to calculate $\mathbb{V}\left(K\left(\frac{t-X}{h}\right) \right) $. To do so, note that 
\begin{align} 
 \mathbb {E} \left(K^2\left(\frac{t-X}{h}\right) \right) &=\int_{-\infty}^{\infty} K^2\left(\frac{t-x}{h}\right)f(x)dx \nonumber \\
&=\int_{-\infty}^{t-ah} f(x)dx + \int_{t-ah}^{t+ah} K^2\left(\frac{t-x}{h}\right)f(x)dx \nonumber \\
&= \int_{-\infty}^{t-ah} f(x)dx  + \int_{t-ah}^{t+ah} \left(\frac{1}{4}+g\left(\frac{t-x}{h}\right)+g^2\left(\frac{t-x}{h}\right)\right)f(x)dx\nonumber \\
&=\frac{3}{4} F(t-ah)+\frac{1}{4}F(t+ah) +  h \int_{-a}^{a} g\left(u\right)f(t-hu)du +  h\int_{-a}^{a} g^2\left(u\right)f(t-hu)du\nonumber \\
&=F(t)-\frac{ah}{2}f(t)+hf(t) \int_{-a}^{a} g^2\left(u\right)du -h^2f^\prime(t)\int_{-a}^{a} ug^2\left(u\right)du+O(h^2)\nonumber \\
&=F(t)-\frac{ah}{2}f(t)+hf(t) \int_{-a}^{a} g^2\left(u\right)du +O(h^2)\nonumber \\
&=F(t)-hf(t)\left(\frac{a}{2}-\int_{-a}^{a}g^2(u)du \right)+O(h^2).\label{exp2ker}
\end{align}
So, using \eqref{expkerO2} and \eqref{exp2ker}, one can write
\begin{align} 
 \mathbb {V} \left(K\left(\frac{t-X}{h}\right) \right) &=  \mathbb {E} \left(K^2\left(\frac{t-X}{h}\right) \right)  -  {\mathbb {E}^2 \left(K\left(\frac{t-X}{h}\right) \right) }\nonumber \\
 &=F(t) \left[1-F(t)\right]-hf(t)\left(\frac{a}{2}-\int_{-a}^{a} g^2(u) du \right)+O(h^2)\nonumber \\
 &=F(t) \left[1-F(t)\right]-hf(t)\left(a-\int_{-a}^{a} K^2(u) du \right)+O(h^2),\nonumber
\end{align} 
where the last equality holds owing to $\int_{-a}^{a} g^2(u) du=\int_{-a}^{a} K^2(u) du-\frac{a}{2}$.
Thus, using the relation \eqref{eqvar2}, the MSE of $F^k_{n;jps}(t)$ is obtained as

\begin{align} 
MSE \left( F^k_{n;jps}(t)\right)&=\mathbb {V}(F^k_{n;jps}(t))+\text{bias}^2 (F^k_{n;jps}(t)) \nonumber \\ 
&\approx H \mathbb {E}(W_1^2 J_1) \left(F(t) \left[1-F(t)\right]-hf(t)\left(a-\int_{-a}^{a} K^2(u) du \right)\right) - r(t)  
+\frac{h^4}{4}  {\left(f^\prime (t)\int_{-a}^{a} u^2k(u)du \right)}^2,\nonumber 
\end{align}
where
$r(t)$ denotes the remaining terms which do not depend on the bandwidth $h$. Therefore, the optimal bandwidth $h$ in the JPS setting can be obtained by minimizing the MSE of $F^k_{n;jps}(t)$ with respect to $h$ as 
\begin{align} 
h_{jps}&={\left(\frac{nH\mathbb {E}(W_1^2 J_1) f(t)\left(a-\int_{-a}^{a} K^2(x) dx\right)}{n {\left( f^\prime (t)\int_{-a}^{a} x^2k(x)dx\right)}^2}\right)}^{\frac{1}{3}} \nonumber \\ 
&={\left(nH\mathbb {E}(W_1^2 J_1)\right)}^{\frac{1}{3}} h_{srs}. \label{hoptjps} \nonumber 
\end{align} 

\begin{remark}
It is clear from equation \eqref{eq8.2} that $nH\mathbb {E}(W_1^2 J_1)\rightarrow 1,$ as $n\rightarrow \infty$, and therefore, the optimal bandwidth $h$ in the JPS setting is almost the same as the SRS setting for large values of $n$.
\end{remark}

\begin{remark}
 We observe that the optimal bandwidth $h$ depends on  quantities $f(t)$ and $f^\prime (t)$ which are often unknown in practice.  We estimate these quantities using \textsf{reference underlying distribution} for these quantities \citep{silverman2018density}. To do so, we take reference distribution as exponential distribution with mean $\bar{X}$ for the situations in which the variable of interest has a positive support such as lifetime data. If the variable of interest has support on real numbers $\mathbb{R}$, then normal distribution with mean $\bar{X}$ and variance $S^2=\frac{1}{n}\sum\limits_{i=1}^n \left(X_i-\bar{X}\right)^2$ is considered as the reference distribution.
\end{remark}

 It is noticeable that  we can obtain kernel-type estimator of PDF  $f$,  from smooth estimator $F^{k}_{n;jps}(t)$  as follows
\begin{eqnarray*}
f^{k}_{n;jps}\left(t\right)= \frac{d}{dt} F^{k}_{n;jps}(t)=\sum \limits_{r=1}^H W_r f^{k}_{n;[r]} (t),
\end{eqnarray*}
where ${f}^{k}_{n;[r]} (t)=\frac{1}{N_r}\sum \limits_{i=1}^n k\left(\frac{t-X_i}{h}\right) I_{ir}$ if $N_r>0$, and zero otherwise.

\section{Monte Carlo Simulation}\label{Sec5}
In this section, we conduct a comprehensive Monte Carlo simulation to compare the performance of the  KDF estimator based on SRS with its JPS sampling counterpart  introduced in section \ref{Sec4},  for both perfect and imperfect ranking set-ups. To this end, we have considered six different distributions as the parent distribution: standard normal distribution ($N(0,1)$), student's t- distribution with 5 degrees of freedom ($t_5$), standard laplace distribution ($La(0,1)$), standard exponential distribution ($E(1)$), gamma distribution with scale parameter 1 and shape parameter 0.5 ($G(0.5, 1)$), and gamma distribution with scale parameter 1 and shape parameter 2 ($G(2, 1)$). Therefore, both distributions with support on the real numbers $\mathbb{R}$ and lifetime distributions with support on $\left(0,+\infty \right)$ are considered in our study.
 
   In this simulation study, we set  $n \in \lbrace 10, 50, 300 \rbrace$, and $H \in \lbrace 3, 5, 10 \rbrace$, therefore, we can compare the performance of the estimators for small, moderate and large values of sample/set size. We can also observe the effect of increasing sample size (set size) on the performance of the estimators while the set size (sample size) is being fixed. 

To perform ranking in the JPS samples, we use the method such as one  utilized  in  linear ranking error model by \citet{dell1972ranked}. 
According to this method, the ranking process is done by an auxiliary variable $(Y)$ correlated with variable of interest $(X)$ as
$$Y=\rho \frac{(X-\mu)}{\sigma}+\sqrt{1-\rho^2}Z,$$
 where $\mu$ and $\sigma^2$ are mean and variance of $X$, respectively, and   the random variable $Z$ is characterized by being independent of $X$, which itself follows a standard normal distribution.  Also, the user controls  quality of ranking by choosing the parameter $\rho\in[0,1]$, as the correlation coefficient between X and Y.    In this simulation study, we select $\rho=1$ for perfect ranking, $\rho=0.9$ for good ranking, $\rho=0.75$ for moderate ranking and $\rho=0.5$ for weak ranking. The kernel functions  utilized in the KDF estimators are epanechnikov kernel function of the form $k(x)=\frac{3}{4} (1-x^2);\ \lvert x  \rvert \le 1$,   triangular kernel function of the form $k(x)=1- \lvert x \rvert;\ \lvert x  \rvert \le 1$,  cosine kernel function of the form $k(x)=\frac{\pi}{4} cos(\frac{\pi}{2}x);\ \lvert x  \rvert \le 1$ and truncated gaussian kernel function of the form  $k(x)=\frac{\phi(x)}{\Phi(4)-\Phi(-4)};\ \lvert x  \rvert \le 4,$   where the functions  $\phi(.)$ and $\Phi(.)$ represent the PDF and CDF of standard normal distribution, respectively.
 
 The relative efficiency (RE) of $F_{n;srs}^k\left(t\right)$ with respect to $F_{n;jps}^k\left(t\right)$ is defined as the ratio of their MSEs, which can be expressed as follows
 $$RE(p)=\frac{MSE(F_{n;srs}^k\left(Q_p\right))}{MSE(F_{n;jps}^k\left(Q_p\right))},$$
 where $Q_p$ is the $p$th quantile of the parent distribution. It is clear that an $RE$ value larger than $1$ indicates an advantage of using  JPS estimator instead of SRS one at the $p$th quantile of the parent distribution. 

For each $\left(n, H, \rho \right)$, we have generated 100,000 random samples from JPS and SRS designs and obtained $RE(p)$ for $p \in \{0.01,0.02, \ldots,0.99\} $. Here, we only present the results in Figures \ref{pic1}-\ref{pic6} for epanechnikov kernel function, because we have observed that  choice of kernel function does not have much effect on the RE of the estimators.

 Figure \ref{pic1} presents the simulation results for sample size $n=10$, and three distributions with support on $\mathbb{R}$. It is clear from this figure that the patterns of the performance of the CDF estimators are almost the same across different distributions. The RE curves  are symmetric and reach their maximum values around their symmetrical point, $p=0.5$. Furthermore, the minimum values of the RE curves are observed at the points close to zero/one.  The first row panels in Figure \ref{pic1} display the outcomes for the perfect ranking scenario.
  It is clear from these panels that the RE curves for $H=3$ and $H=5$ are almost identical and are slightly higher than that of  for $H=10$. This can be justified by the fact that empty strata are frequently observed for $\left(n,H\right)=\left(10,10\right)$, so the post stratification with set size $H=10$ does not contribute much in improving estimation precision. Furthermore, the REs are higher than one except for two narrow intervals at the boundaries of the parent distribution. The simulation results for good ranking case ($\rho=0.9$) are depicted in the second row of Figure \ref{pic1}. From these panels, it can be observed that the patterns of the RE  curves closely resemble those of the perfect ranking case ($\rho=1$),  with a clear difference that the REs are lower in this case and the span of intervals in which the REs are below one become wider. Simulation results for  $\rho=0.75$, and $\rho=0.5$ are demonstrated in the third and the bottom rows of the Figure \ref{pic1}, respectively. It is evident that although the patterns of RE curves are the same as  those  of $\rho=0.9$,  RE values decrease as the value of $\rho$ decreases. Specifically, for $\rho=0.5$, the REs are below one except for a narrow interval around $p=0.5$. 
 
 The simulation results for sample size $n=50$, and three distributions with support on $\mathbb{R}$ are shown in Figure \ref{pic2}. It is obvious that the  RE curves have symmetric forms and are almost the same across different distributions. The three top panels of Figure \ref{pic2} show the results for $\rho=1$. It is evident from these panels that the efficiencies of the JPS estimators are increasing function in set size $H$, and thus RE curve for $H=10$ ($H=5$) is uniformly higher than that for $H=5$ ($H=3$). Each RE curve reaches its maximum for $p=0.5$ and minimum for $p=0.01$ and $p=0.99$. Furthermore, it has two local minimum points around $p=0.4$ and $p=0.6$, and falls below one for a very narrow interval close to zero/one. The panels in the second top row of Figure \ref{pic2} demonstrate the results for $\rho=0.9$. The patterns of the performances of JPS estimators are almost the same as those for $\rho=1$, with a obvious difference that the REs are lower in this case. We can observe the performances of the JPS estimators for $\rho=0.75$ in the panels of the third row of Figure \ref{pic2}. There are two main differences among the RE curves in this case and those for $\rho=0.9$. First, the RE curves are lower and the spans of the intervals in the REs are below of one are wider  for $\rho=0.75$ than $\rho=0.9$. Second, the RE curves are not increasing function in $H$, any more.  These differences are also observed when we compare the results for $\rho=0.5$ in the panels in the bottom row of Figure \ref{pic2} with those for $\rho=0.75$. It is of interest to note that the JPS estimator with $H=10$ has the lowest efficiency in this case.

  The simulation results for sample size $n=300$, and three distributions with support on $\mathbb{R}$ are shown in Figure \ref{pic3}. Thus, from this figure, we can observe the asymptotic performance of the estimators. It is clear that REs increase  (decrease) when the set size $H$ (ranking quality $\rho$) increases (decreases) while the other parameters are kept fixed. It is also worth mentioning that the REs never fall below one for $n=300$. 

Figures \ref{pic4}-\ref{pic6} present simulation results for lifetime distributions for sample sizes $n=10, 50$ and $300$, respectively. Figure \ref{pic4} shows the simulation results for lifetime distributions for $n=10$.  Based on the presented figure, it is evident 
 that the RE curves of the JPS estimators resemble different shapes for different distributions. From the panels in the first row of Figure \ref{pic4}, we observe the JPS estimator for $H=5$ ($H=10$) usually has the best (worst) performance. The RE curves are not symmetric, but they reach to their maximum around the center of the distribution and  (local minimum) minimum around the (lower) upper tail of the parent distribution which is below one.  From the panels in the second row of Figure \ref{pic4}, we observe the patterns of the performances of the JPS estimators for $\rho=0.9$ are almost the same as those for $\rho=1$, but they are lower. The simulation results are not in favor of the JPS estimators for $\rho=0.75$ and $\rho=0.5$, which are  shown in the third and bottom rows of Figure \ref{pic4}, respectively, as their REs are below one for most values of $p$. 

Simulation results for lifetime distributions when the sample size is $n=50$ are depicted in Figure \ref{pic5}. Based on the panels displayed in the first row of the provided figure, we can observe that in cases where the ranking process is perfect ($\rho=1$), the JPS estimator with $H=10$ performs better than the others in most considered cases, which is followed by the JPS estimators with $H=5$, and $H=3$, respectively. However, this pattern does not hold for $\rho=0.9$, as we see it in the second row of Figure \ref{pic5}, and the results for $H=5$ and $H=10$ are quite competitive. The simulation results for $\rho=0.75$ are not in favor of $H=10$ as shown in the third row of Figure \ref{pic5}.  In fact, the JPS estimator with $H=5$ has the best performance for most values of $p$, which is followed by the estimator with $H=3$. The three bottom panels of Figure \ref{pic5} show  the results for $\rho=0.5$. It is clear from these panels that the REs are close or below one, and the JPS estimator with a smaller set size usually has a better performance in this case. 

Simulation results for lifetime distributions and $n=300$ are shown in Figure \ref{pic6}.  From the provided figure, it can be observed that the  REs of the estimators increase as the set size $H$ becomes larger, and the JPS estimator with set size $H=10$ has the best performance in most considered case, which is followed by the estimators with set size $H=5$, and $H=3$, respectively. It is also worth mentioning that the efficiency of the JPS estimators decrease as the ranking quality decreases. Furthermore, the RE curves are mostly  higher than one, which indicate superiority of the JPS estimators to their SRS counterpart.

By overall examinations of Figures \ref{pic1}-\ref{pic6}, we find out that the REs are increasing function of sample size $n$. The efficiency of the JPS estimators are usually higher around the center of the parent distribution and lower at its boundaries. Furthermore, the REs decrease with the ranking quality $\rho$. The optimal value of  the set size $H$, which leads to a higher efficiency  depends on both ranking quality ($\rho$) and sample size ($n$). Specifically, a larger value of set size $H$ leads to a higher efficiency provided that the sample size is \textsf{large enough} and the ranking quality is \textsf{sufficiently good}. Therefore, small values for set sizes ($H=3, 5$) are recommended to be used in practice if the sample size is small or there is any doubt about ranking quality. 
  
\begin{figure}
\begin{center}
\includegraphics[scale=0.9]{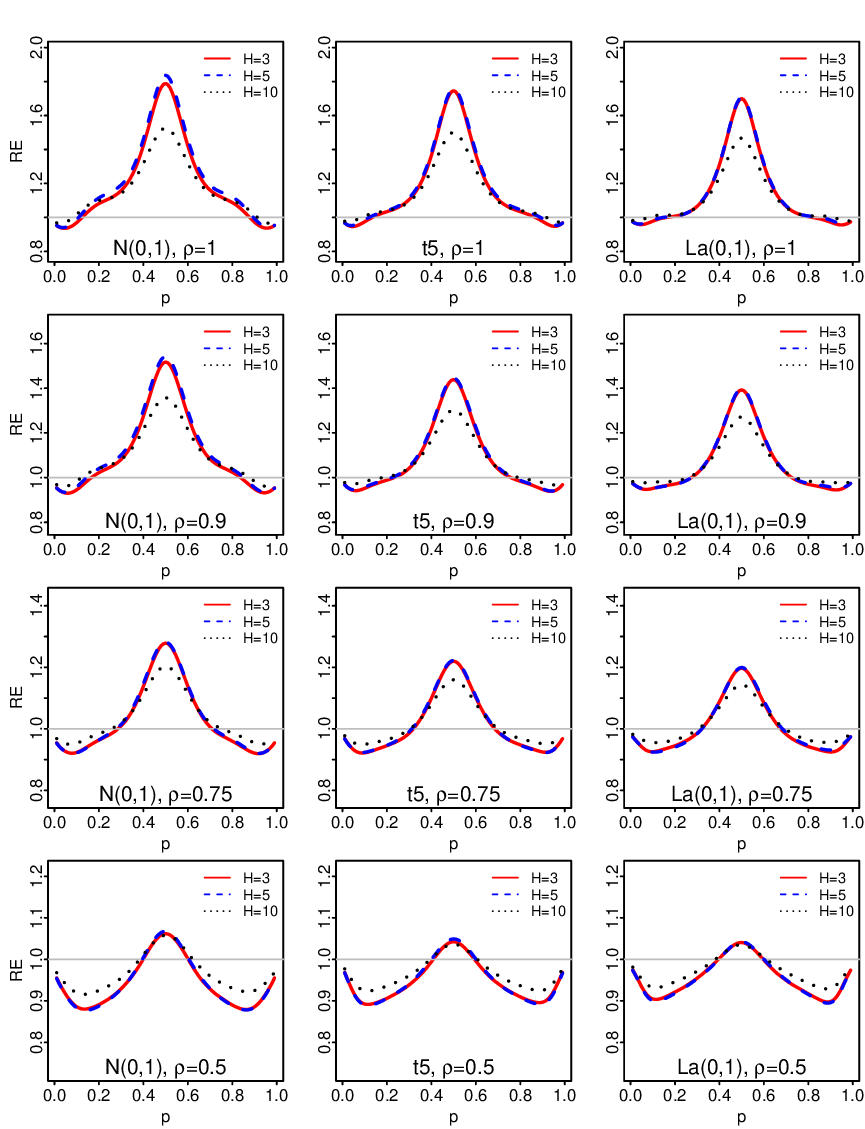}
\caption{\label{pic1}The estimated relative efficiency of $F_{n;srs}^k(.)$  to $ F_{n;jps}^{k}(.)$  as a function of $p$,  for $n=10$, when $H=3$ (represented by red and solid line), $H=5$ (represented by blue and dashed line), $H=10$ (represented by black and dotted line) and  $\rho \in \{1,0.9,0.75,0.5\}$ under $ N(0,1), t5$  and $La(0,1)$  distributions. }
\end{center}
\end{figure} 

\begin{figure}
\begin{center}
\includegraphics[scale=0.9]{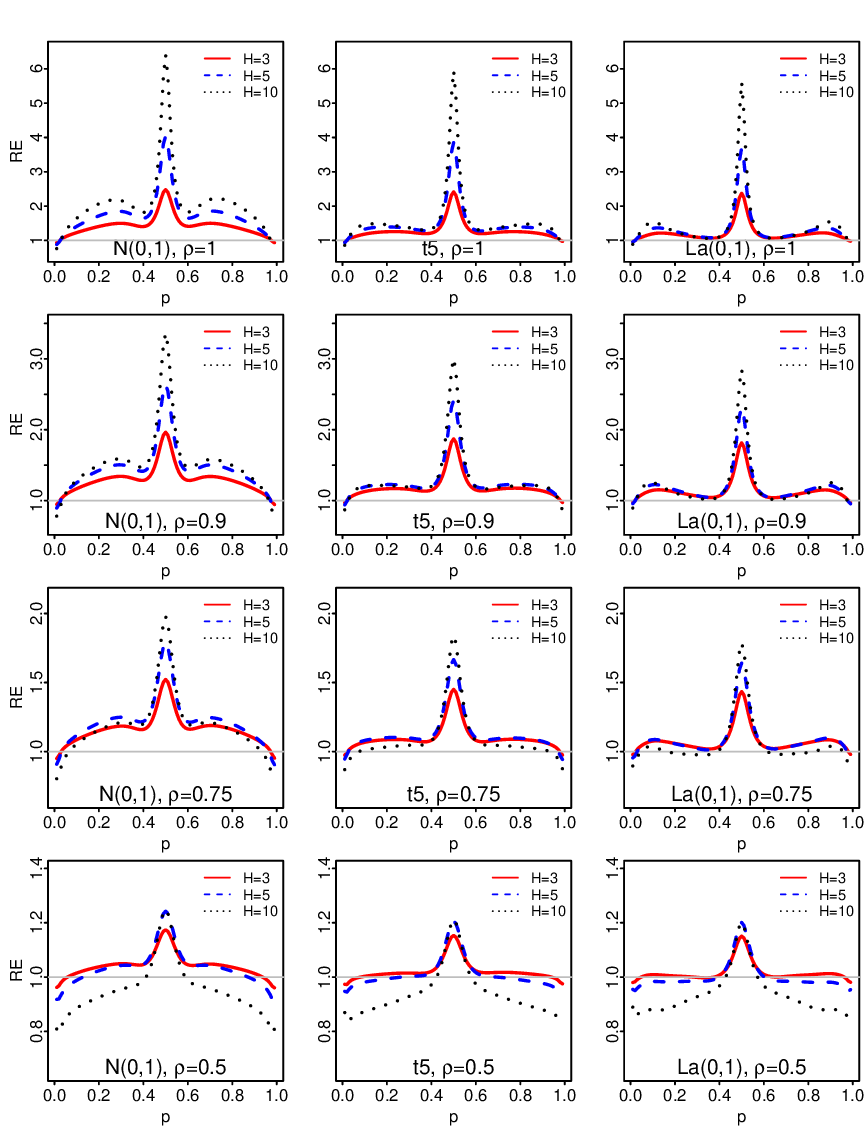}
\caption{\label{pic2}The estimated relative efficiency of $F_{n;srs}^k(.)$  to $ F_{n;jps}^{k}(.)$  as a function of $p$,  for $n=50$, when $H=3$ (represented by red and solid line), $H=5$ (represented by blue and dashed line), $H=10$ (represented by black and dotted line) and  $\rho \in \{1,0.9,0.75,0.5\}$ under $ N(0,1), t5$  and $La(0,1)$  distributions.}
\end{center}
\end{figure} 

\begin{figure}
\begin{center}
\includegraphics[scale=0.9]{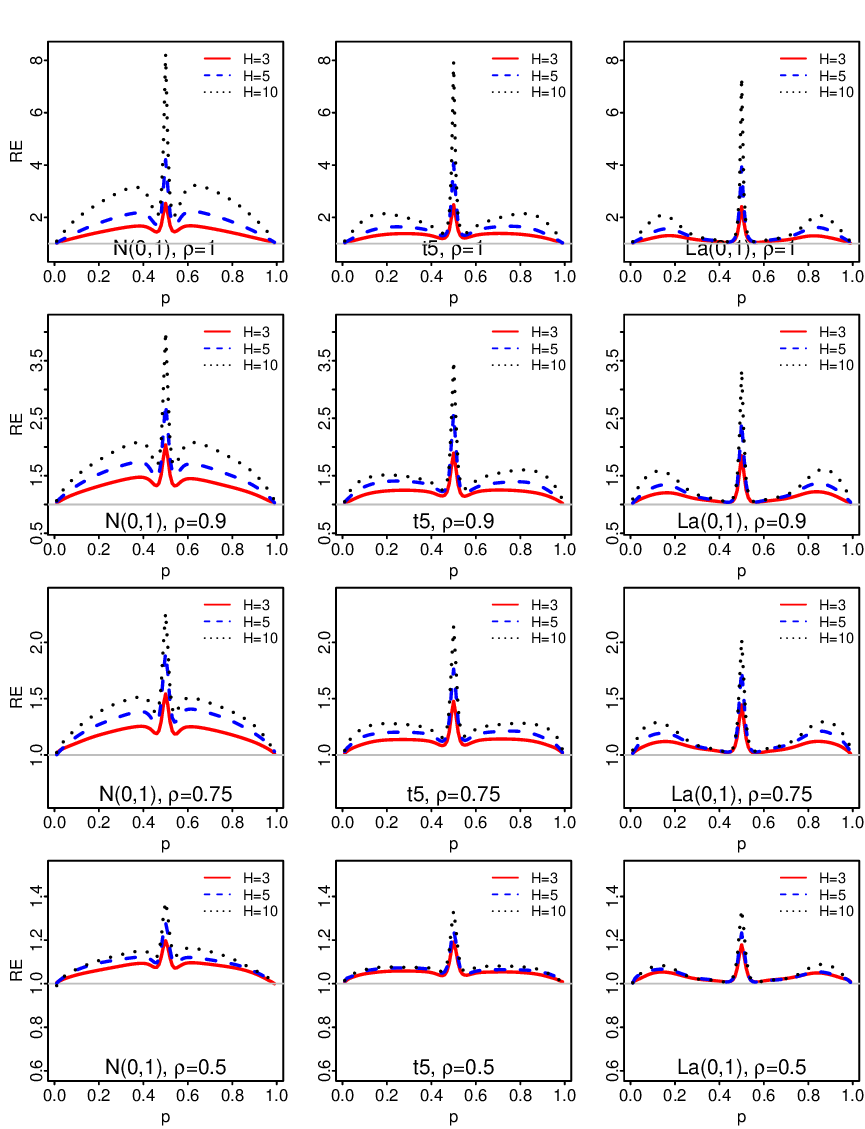}
\caption{\label{pic3}The estimated relative efficiency of $F_{n;srs}^k(.)$  to $ F_{n;jps}^{k}(.)$  as a function of $p$,  for $n=300$, when $H=3$ (represented by red and solid line), $H=5$ (represented by blue and dashed line), $H=10$ (represented by black and dotted line) and  $\rho \in \{1,0.9,0.75,0.5\}$ under $ N(0,1), t5$  and $La(0,1)$  distributions.}
\end{center}
\end{figure}

 \begin{figure}
\begin{center}
\includegraphics[scale=0.9]{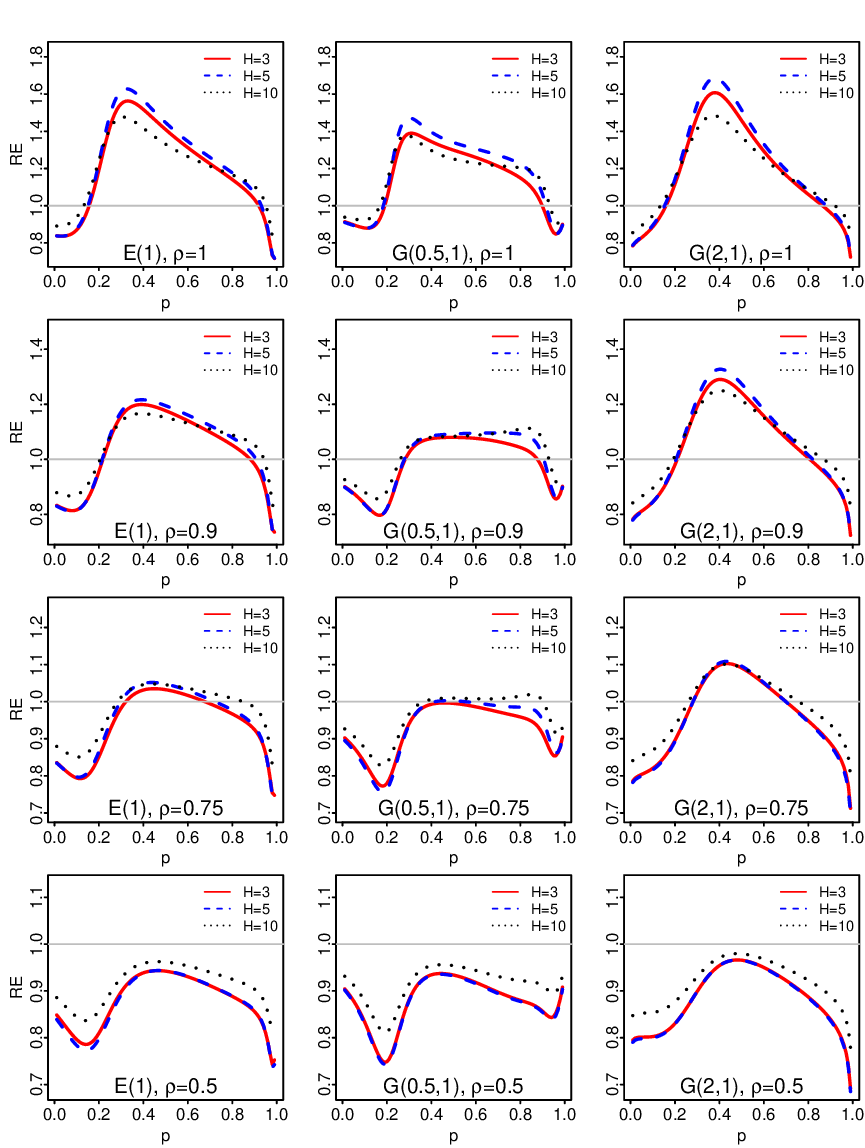}
\caption{\label{pic4}The estimated relative efficiency of $F_{n;srs}^k(.)$  to $ F_{n;jps}^{k}(.)$  as a function of $p$,  for $n=10$, when $H=3$ (represented by red and solid line), $H=5$ (represented by blue and dashed line), $H=10$ (represented by black and dotted line) and  $\rho \in \{1,0.9,0.75,0.5\}$ under $ E(1), G(0.5,1)$  and $G(2,1)$  distributions.}
\end{center}
\end{figure} 

 \begin{figure}
\begin{center}
\includegraphics[scale=0.9]{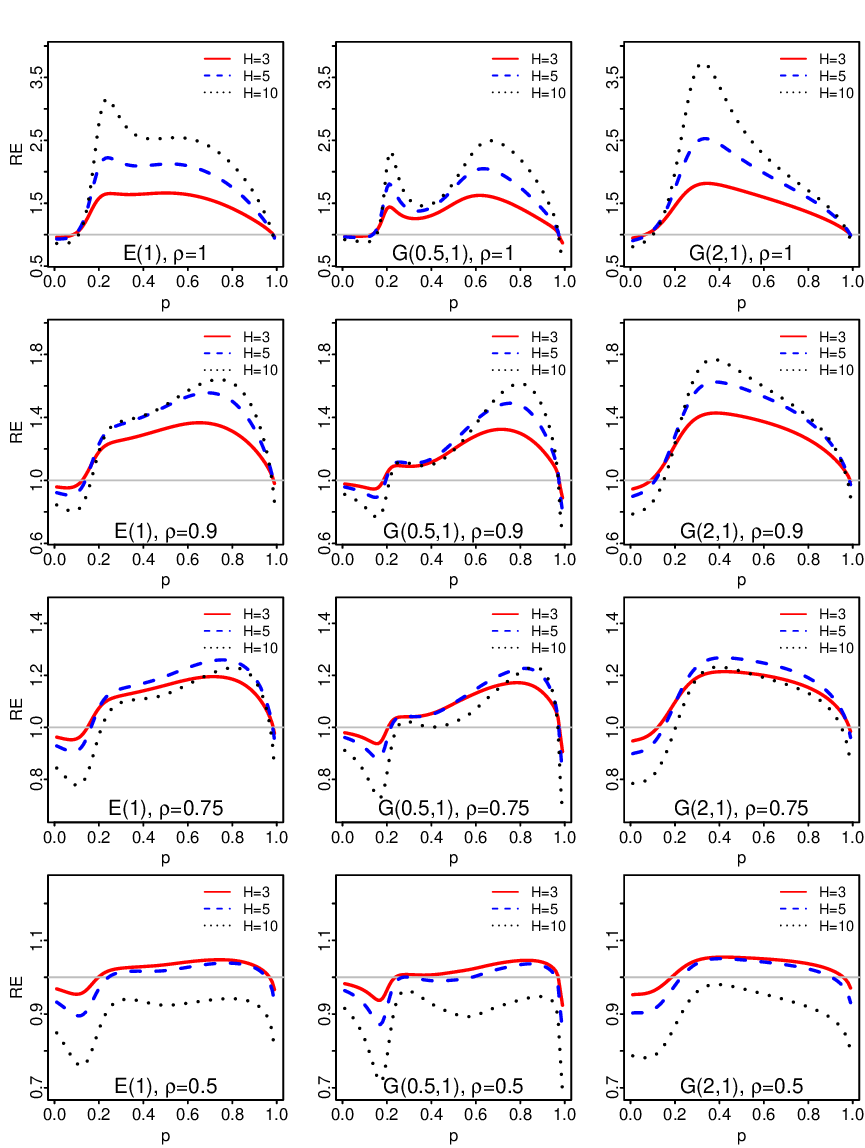}
\caption{\label{pic5}The estimated relative efficiency of $F_{n;srs}^k(.)$  to $ F_{n;jps}^{k}(.)$  as a function of $p$,  for $n=50$, when $H=3$ (represented by red and solid line), $H=5$ (represented by blue and dashed line), $H=10$ (represented by black and dotted line) and  $\rho \in \{1,0.9,0.75,0.5\}$ under $ E(1), G(0.5,1)$  and $G(2,1)$  distributions.}
\end{center}
\end{figure} 

 \begin{figure}
\begin{center}
\includegraphics[scale=0.9]{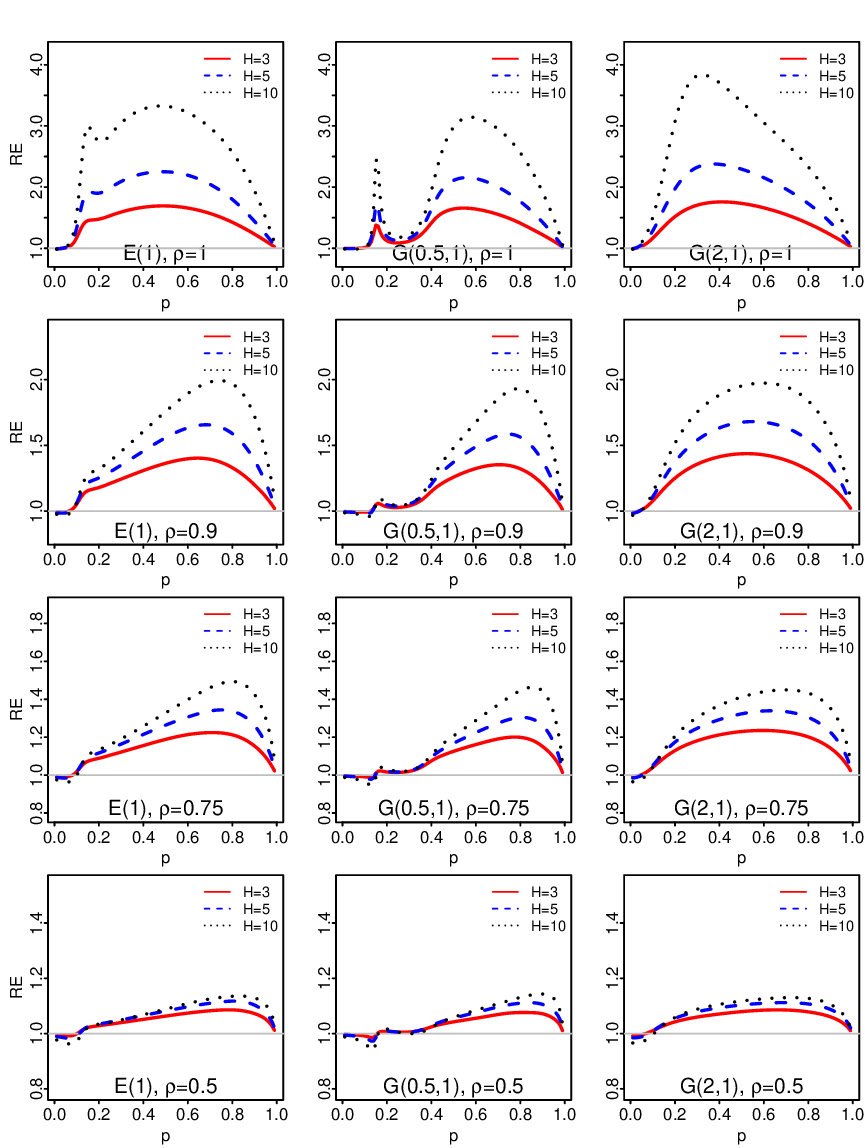}
\caption{\label{pic6}The estimated relative efficiency of $F_{n;srs}^k(.)$  to $ F_{n;jps}^{k}(.)$  as a function of $p$,  for $n=300$, when $H=3$ (represented by red and solid line), $H=5$ (represented by blue and dashed line), $H=10$ (represented by black and dotted line) and  $\rho \in \{1,0.9,0.75,0.5\}$ under $ E(1), G(0.5,1)$  and $G(2,1)$  distributions.}
\end{center}
\end{figure} 

\section{An Empirical Study: Body Fat Estimation}\label{Sec6}

Our purpose in this section is to demonstrate the applicability and efficiency of the introduced procedure in practice using a real dataset. The dataset is utilized to evaluate the performance of the CDF estimators in JPS and SRS settings  where the ranking process is  carried out using actual auxiliary variables, rather than ones generated by a model.

It is very important for organizations in health sector such as Food and Drug Administration (FDA), or the National Cancer Institute (NCI) to have an accurate estimation of the distribution of body fat of people in a given population. This estimate is very useful for these organizations,  since it can provide insights into the population's risk for some non-communicable illness such as diabetes, cardiovascular diseases,  and some certain types of cancers, which usually impose a huge financial burden on the health care system of the society.  However, the exact quantifying a person's body fat is a very challenging job and requires some advanced technologies. This is mainly because the body fat has not distributed evenly throughout the body. The exact measurement of a person's body fat is often obtained using some advanced imaging techniques such as Dual-Energy X-ray Absorptiometry (DEXA), Magnetic Resonance Imaging (MRI), or Computerized Tomography (CT) scans. Note that although these methods are very precise, they are very expensive as well, and need some special equipment and expertise which limit their applicability for large-scale population studies or even a standard clinical practice. 
 The body fat percentage variable is correlated with some easily available concomitant variables such as  abdomen circumference, chest circumference and weight.

The dataset utilized in this section is known as \textsf{bodyfat} dataset, and is accessible online at \href{https://lib.stat.cmu.edu/datasets/bodyfat}{https://lib.stat.cmu.edu/datasets/bodyfat}\footnote{Access date: 6 April 2024.}. This dataset contains body fat measurements determined using underwater weighting technique for 252 men along with measurements of some of their body circumferences. In what follows, we assume that this dataset represents our \textsf{hypothetical} population, and it is utilized to conduct a comparison of the performance exhibited by various CDF estimators. Assume that our objective is to estimate the true CDF of body fat among individuals,  which is given by 
$$
F\left(t\right)=\frac{1}{252}\sum_{i=1}^{252} \mathbb{I} \left(X_i\leq t \right),
$$
where $X_i$ is the body fat measurement of the $i$th subject in the hypothetical population. 
 
To compare the CDF estimators in the JPS and SRS settings, we set $n\in \lbrace 10, 20, 30, 40, 50 \rbrace$, $H\in \lbrace 3, 5, 10 \rbrace$, and for each combination of $\left(n, H\right)$, we have drawn 100,000 samples from both of JPS and SRS designs, where all samplings are considered with replacement and therefore the independence assumption is not violated. 

 To obtain a JPS sample, we assume that ranking process is carried out using one of the concomitant variables, abdomen circumference, chest circumference, and weight. The correlation coefficient between the variable of interest (body fat) and the concomitant variables are $\rho=0.81$, $\rho=0.70$, and $\rho=0.61$, respectively. 
 
 To include the perfect ranking case ($\rho=1$) in our study, we also use the variable of interest for ranking itself. We have utilized  epanechnikov  function for the kernel function, and compute the RE of the JPS estimator to its SRS counterpart as it is defined in Section \ref{Sec5} for $p \in \lbrace 0.1, 0.25, 0.5, 0.75, 0.9 \rbrace$. The results are given in Table \ref{Table1}.

The first part of the Table \ref{Table1} gives the results when the variable of the interest is ranked by itself, so the perfect ranking assumption holds. From the information presented in these rows, it is evident that one of the estimators derived from the JPS sampling outperforms the others and can be considered the superior estimator. 
 except for the cases that $n\leq 20$ and $p=0.1, 0.9$. The best estimator  in each case depends on the  sample size $n$, and the value of $p$. However, we observe that the performance of the JPS estimator improves as the value of $p$ approaches to $p=0.5$, or sample size $n$ increases. For $p=0.1$ or $p=0.9$, the best JPS estimator is mostly observed for  $H=3$. But, for $p \in\lbrace 0.25, 0.5, 0.75 \rbrace$, the JPS estimator with either $H=5$ or $H=10$ has the best performance depending on the magnitude of the sample size $n$. 

The other parts of the Table \ref{Table1} show the results for imperfect ranking case. We observe that although the general patterns of REs remain almost the same, the REs decrease as the value of $\rho$ decreases. Furthermore, the number of cases in which the JPS estimator cannot beat the SRS one increases as the value of $\rho$ reduces.  All of these observations align with the results presented in Section \ref{Sec5}.

\begin{landscape}
\begin{table}[]
\scriptsize
\centering
\begin{tabular}{l cl c l c l c l c l c l c l c l c l c l c  l c l c  lc l c  l c l c   l c  l c  lc l c  l}\hline
Concomitant &    \multicolumn{3}{c}{$p=0.1$}       &  &                \multicolumn{3}{c}{$p=0.25$}   &  &                \multicolumn{3}{c}{$p=0.5$}       &  &                \multicolumn{3}{c}{$p=0.75$}   &  &                \multicolumn{3}{c}{$p=0.9$}   \\
\cline{3-5}       \cline{7-9}    \cline{11-13}     \cline{15-17}      \cline{19-21}
Variable      & $n$ & $H=3$& $H=5$ & $H=10$  & & $H=3$ &  $H=5$   &  $H=10$ & & $H=3$ &  $H=5$   &  $H=10$ & & $H=3$ &  $H=5$   &  $H=10$ & & $H=3$ &  $H=5$   &  $H=10$  \\\hline\hline
                                 & 10    & 0.89 & 0.89 & 0.94    & &      1.09 & 1.11 & 1.12    & &    1.62 & 1.71& 1.48    & &    1.00 & 1.00 & 1.01    & &    0.90 & 0.90 & 0.93       \\ 
Body Fat                    & 20    & 0.98 &  0.92 & 0.91    & &     1.32 & 1.35 & 1.32    & &    1.84 & 2.34 & 2.21    & &   1.09 & 1.05 & 1.03    & &    0.96 & 0.89 & 0.88       \\ 
($\rho=1$)               & 30    &  1.04 & 1.00 &0.94   & &      1.47 &1.65 &  1.64    & &    1.90 & 2.65 & 3.01   & &  1.15 & 1.15& 1.07       & &     1.00 & 0.95 &0.86        \\ 
                                 & 40    & 1.07 &1.07 & 0.98   & &     1.55 &1.89 & 1.99   & &   1.93 &  2.75 & 3.63    & &    1.19 & 1.22 & 1.14    & &    1.02 & 0.98 & 0.88        \\ 
                                 & 50    &1.09 & \textbf{1.11} & 1.03   & &     1.61 & 2.04 & \textbf{2.38}    & &    1.93 & 2.76 & \textbf{4.15}    & &    1.21 &\textbf{1.26} & 1.21    & &    \textbf{1.03} & 1.01 &  0.91      \\            \\\hline \hline

                                & 10    & 0.85 & 0.84 & 0.90    & &     0.97 & 0.97 & 1.00    & &    1.28 & 1.32 & 1.24    & &    0.95 & 0.95 & 0.97    & &    0.88 & 0.88 & 0.91       \\ 
Abdomen                & 20    & 0.94 & 0.86 & 0.84    & &     1.13 & 1.08 & 1.05    & &    1.39 & 1.52 & 1.46    & &    1.02 & 0.96 & 0.95    & &    0.94 & 0.86 & 0.85       \\ 
Circumference         & 30    & 0.99 & 0.92 & 0.86    & &     1.22 & 1.23 & 1.16    & &    1.43 & 1.61 & 1.62    & &    1.07 & 1.03 & 0.95    & &    0.98 & 0.91 & 0.82       \\ 
($\rho=0.81$)                                & 40    & 1.02 & 0.98 & 0.87   & &     1.26 & 1.35 & 1.27    & &    1.43 & 1.65 & 1.74   & &    1.09 & 1.09 & 1.00   & &    1.00 & 0.95 & 0.84       \\ 
                                & 50    &\textbf{1.04} & 1.02 & 0.91    & &     1.31& \textbf{1.42} & 1.39    & &    1.45 & 1.67 & \textbf{1.81}    & &   1.11 & \textbf{1.12} & 1.04   & &   \textbf{1.01} &0.98 & 0.86       \\               \\\hline \hline

                               & 10    & 0.83 & 0.82 & 0.88    & &     0.90 & 0.91 & 0.94    & &    1.13 & 1.15 & 1.12    & &    0.92 & 0.92 & 0.95    & &    0.87 & 0.87 & 0.91       \\ 
Chest                      & 20    & 0.92 & 0.82 & 0.81    & &     1.03 & 0.95 & 0.93    & &    1.21 & 1.23 & 1.21    & &    0.98 & 0.92 & 0.90    & &    0.92 & 0.85 & 0.83       \\ 
Circumference          & 30    &  0.96 & 0.89 & 0.80    & &    1.10 & 1.06 & 0.96    & &    1.23 & 1.28 & 1.25    & &    1.03 & 0.97 & 0.89    & &    0.96 & 0.89 & 0.81       \\ 
($\rho=0.7$)                               & 40    & 0.99 & 0.95 & 0.82   & &     1.14 & 1.14 & 1.03    & &    1.24 & 1.32 & 1.29    & &    1.05 & 1.02 & 0.92    & &    0.98 & 0.93 & 0.81       \\ 
                               & 50    & \textbf{1.01} & 0.97 & 0.86    & &     1.16 &\textbf{1.18} & 1.10    & &    1.24 & \textbf{1.33} & 1.32   & &    1.05 & \textbf{1.06} & 0.95    & & {0.99} & 0.95 & 0.83       \\                \\\hline\hline           

                              & 10    & 0.82 & 0.82 & 0.87    & &     0.89 & 0.88 & 0.93    & &    1.08 & 1.09 & 1.08    & &    0.91 & 0.90 & 0.94    & &    0.87 & 0.86 & 0.90       \\ 
Weight                   & 20    & 0.90 & 0.81 & 0.80    & &     0.99 & 0.91 & 0.90    & &    1.14 & 1.14 & 1.12    & &    0.97 & 0.90 & 0.88    & &    0.92 & 0.84 & 0.82       \\ 
($\rho=0.61$)       & 30    & 0.96 & 0.88 & 0.79    & &     1.07 & 1.02 & 0.93    & &    1.17 & 1.19 & 1.16    & &    1.01 & 0.95 &0.87     & &    0.96 & 0.88 & 0.80       \\ 
                              & 40    & 0.98 & 0.93 & 0.82    & &     1.10 & 1.09 & 0.99    & &    1.18 & 1.22 & 1.19   & &    1.03 & 0.99 & 0.89    & &    0.97 & 0.92 & 0.80       \\ 
                              & 50    & 0.99 & 0.96 & 0.84    & &     1.12 & \textbf{1.14} & 1.04    & &   1.19 &\textbf{1.25} & 1.22    & &    \textbf{1.04} & 1.02 & 0.93    & &    0.98 & 0.94 & 0.83       \\                    \\\hline         
\end{tabular}
\caption{The estimated relative efficiency of $F_{n;srs}^k(.)$  to $ F_{n;jps}^{k}(.)$  using the bodyfat dataset.  In each scenario, the winner is indicated by bold font.}
\label{Table1}
\end{table}
\end{landscape}

\section{Discussion}\label{Sec7}
Judgment post stratification (JPS) sampling plan is beneficial sampling plan for situations in which allocating a judgment rank to a sample unit in a set is far easier/cheaper than   precisely  quantifying it. It is often considered as a more flexible and practical variation of ranked set sampling (RSS).  Despite the extensive utilization of this sampling approach in various fields, the literature lacks discussions on the issue of effectively smooth estimating the cumulative distribution function (CDF).

In this paper, we discussed a general class of the estimators for the population CDF in the JPS setting which includes both empirical and kernel-based ones.  We showed that they are  more efficient than their competitors in simple random sampling (SRS), as long as the ranking quality is better than random guessing. We found the necessary and sufficient condition for the estimators to be asymptotically unbiased. Also, we studied the  Glivenko-Cantelli type convergence and  asymptotic normality for the estimators in this class, assuming the same condition holds.
We next focused on kernel distribution function (KDF) in the class and proposed optimal bandwidth.  We conducted a comparative analysis between the performance of the KDF  in JPS sampling  and its competitor in SRS  using Monte Carlo simulation. This analysis involved investigating various combinations of sample size, set size, ranking quality, parent distribution, and kernel function.

 It was found that the JPS estimator outperforms its SRS competitor in the most considered cases. Finally, we showed the applicability and efficiency of our introduced procedure in practice using a real dataset in which real concomitant variables were used for the ranking process.

To the best of our knowledge, this work was the first attempt for estimating the CDF using kernel function based on JPS sampling scheme. So, There is still plenty of room for more research.  As an example, consider a parameter of interest denoted by $\theta$, which can be expressed as $\theta = g(F)$, where $g(.)$ represents a known function.
 One can think of estimating the parameter $\theta$ by replacing the CDF $F$ by an appropriate estimator from the discussed class, and establishing its statistical properties. Moreover, it was shown that the true CDF of the post strata often follow the constraint: $F_{[1]}\left(t\right)\geq\ldots  \geq F_{[H]}\left(t\right)$, for all $t\in \mathbb{R}$. However, this constraint may not hold for sample estimates. \cite{wang2012isotonized} improves the EDF in the JPS setting by imposing this limitation onto estimation process. Therefore, another interesting topic for future research can be how to utilize this limitation to improve the performance of  the KDF in the JPS setting. These topics can be discussed in  future works.

\section*{ Data Availability Statement}
The data supporting the findings of this study can be accessed online at  \href{https://lib.stat.cmu.edu/datasets/bodyfat}{https://lib.stat.cmu.edu/datasets/bodyfat} (Access date: 6 April 2024).

\newpage
\newpage
\bibliographystyle{plainnat}     
\bibliography{ref2.bib}

\end{document}